\documentclass{article}

\newcommand{\cmp}{Comm. Math. Phys.~}

\newcommand{\jpa}{J. Phys. A~}

\newcommand{\prl}{Phys. Rev. Lett.~}
\newcommand{\pra}{Phys. Rev. A~}

\usepackage{color}
\usepackage{subfigure}
\usepackage[T1]{fontenc}
\usepackage[sc]{mathpazo}
\usepackage{amsmath}
\usepackage{amssymb}
\usepackage{graphicx,color}
\usepackage{framed}
\usepackage{multirow}
\usepackage{enumerate}
\usepackage{amsthm}
\usepackage{amsfonts,mathrsfs}
\usepackage{geometry} 
\usepackage{eepic}
\usepackage{ifthen}
\usepackage[vcentermath]{youngtab}
\usepackage[unicode=true,pdfusetitle, bookmarks=true,bookmarksnumbered=false,bookmarksopen=false, breaklinks=false,pdfborder={0 0 0},backref=false,colorlinks=false] {hyperref}
\hypersetup{
colorlinks,linkcolor=myurlcolor,citecolor=myurlcolor,urlcolor=myurlcolor}
\definecolor{myurlcolor}{rgb}{0,0,0.7}

\usepackage{color}

\usepackage{hyperref}
\hypersetup{pdfpagemode=UseNone}

\newcommand{\tinyspace}{\mspace{1mu}}

\newcommand{\op}[1]{\operatorname{#1}}

\newcommand{\abs}[1]{\left\lvert\tinyspace #1 \tinyspace\right\rvert}

\newcommand{\norm}[1]{\left\lVert\tinyspace #1 \tinyspace\right\rVert}

\renewcommand{\det}{\operatorname{det}}
\renewcommand{\t}{{\scriptscriptstyle\mathsf{T}}}

\newcommand{\setft}[1]{\mathrm{#1}}

\newcommand{\density}[1]{\setft{D}\left(#1\right)}

\renewcommand{\vec}{\op{vec}}

\def\u{\mathfrak{u}}

\def\k{\mathfrak{k}}

\def\liet{\mathfrak{t}}

\def\Ad{\mathrm{Ad}}

\def\vol{\mathrm{vol}}

\def \dif {\mathrm{d}}
\def \diag {\mathrm{diag}}
\def \vol {\mathrm{vol}}

\def\complex{\mathbb{C}}
\def\real{\mathbb{R}}

\def\I{\mathbb{1}}

\newenvironment{mylist}[1]{\begin{list}{}{
    \setlength{\leftmargin}{#1}
    \setlength{\rightmargin}{0mm}
    \setlength{\labelsep}{2mm}
    \setlength{\labelwidth}{8mm}
    \setlength{\itemsep}{0mm}}}
    {\end{list}}

\def\ot{\otimes}

\newcommand{\iinner}[2]{\langle #1 | #2\rangle}
\newcommand{\out}[2]{| #1\rangle\langle #2 |}
\newcommand{\Inner}[2]{\left\langle #1 , #2\right\rangle}
\newcommand{\Innerm}[3]{\left\langle #1 \left| #2 \right| #3 \right\rangle}

\newcommand{\pa}[1]{(#1)}
\newcommand{\Pa}[1]{\left(#1\right)}

\newcommand{\Br}[1]{\left[#1\right]}
\newcommand{\set}[1]{\{#1\}}
\newcommand{\Set}[1]{\left\{#1\right\}}

\newcommand{\ket}[1]{|#1\rangle}

\DeclareMathOperator{\trace}{Tr}
\newcommand{\ptr}[2]{\trace_{#1}\pa{#2}}
\newcommand{\Ptr}[2]{\trace_{#1}\Pa{#2}}

\newcommand{\Tr}[1]{\Ptr{}{#1}}

\def\cC{\mathcal{C}}
\def\cH{\mathcal{H}}
\def\cK{\mathcal{K}}\def\cO{\mathcal{O}}

\def\bP{\mathbf{P}}

\def\bsM{\boldsymbol{M}}

\def\bsr{\boldsymbol{r}}\def\bss{\boldsymbol{s}}

\def\rD{\mathrm{D}}
\def\rG{\mathrm{G}}\def\rH{\mathrm{H}}
\def\rL{\mathrm{L}}\def\rO{\mathrm{O}}
\def\rS{\mathrm{S}}
\def\rU{\mathrm{U}}

\def\sD{\mathscr{D}}

\def\sL{\mathscr{L}}

\newtheorem{thrm}{Theorem}[section]
\newtheorem{lem}[thrm]{Lemma}
\newtheorem{prop}[thrm]{Proposition}
\newtheorem{cor}[thrm]{Corollary}
\theoremstyle{definition}

\newtheorem{remark}[thrm]{Remark}
\newtheorem{exam}[thrm]{Example}

\numberwithin{equation}{section}

\newcounter{questionnumber}
\newenvironment{question}
       {\begin{mylist}{\parindent}
     \item[\stepcounter{questionnumber}\thequestionnumber.]}
       {\end{mylist}}

\begin{document}

\title{Volume of the set of locally diagonalizable bipartite states}

\author{Lin Zhang$^1$\footnote{E-mail: godyalin@163.com; linyz@hdu.edu.cn},\quad Seunghun
Hong$^2$\footnote{E-mail: seunghun.hong@gmail.com}\\
  {\it\small $^1$Institute of Mathematics, Hangzhou Dianzi University, Hangzhou 310018, PR~China}\\
  {\it\small $^2$Northwestern College, 101 7th St SW, Orange City, 51041 IA, USA}}
\date{}
\maketitle
\maketitle \mbox{}\hrule\mbox\\
\begin{abstract}

The purpose of this article is to investigate the geometry of the
set of locally diagonalizable bipartite quantum states. We have
the following new results: the Hilbert-Schmidt volume of all locally
diagonalizable states, and a necessary and sufficient condition for
local diagonalizability in the qubit-qubit case. Besides, we
partition the set of all locally diagonalizable states as local
unitary orbits (or coadjoint orbits) of diagonal forms. It is
well-known that the Riemannian volume of a coadjoint orbit for a
regular point in a specified Weyl chamber can be calculated
by Harish-Chandra's volume formula. By modifying Harish-Chandra's
volume formula, we give, for the first time, a specific
formula for the Riemannian volume of a local unitary
orbit of a regular point in a specified Weyl chamber.
Several open questions are presented as well.

\end{abstract}
\maketitle \mbox{}\hrule\mbox

\section{Introduction}

Qubits and qubit quantum channels are the simplest building blocks
for quantum information processing and quantum computations. A qubit
is the quantum analog of the classical bit; a qubit quantum channel
is just the quantum analog of the transition probability matrix.
Recently, Lovas and Andai \cite{Lovas2016a} analyzed the structure
of these qubit channels using the duality between quantum maps and
quantum states, i.e., via Choi-Jami{\l}kowski correspondence
\cite{Karol2004}. They calculate the (Euclid) volume of general and
unital qubit channels (real and complex case) with respect to the
Lebesgue measure. For unital qubit channels, they are essentially
equivalent to two-qubit states with the same completely-mixed
marginal states via Choi-Jami{\l}kowski representation.


In the recent decades, the geometric separability probability of
bipartite systems, i.e., the ratio of volumes of the set of
separable bipartite states to the set of all bipartite states on the
same tensor space of two Hilbert spaces, has been extensively
studied. In 1998, \.{Z}yczkowski \emph{et al} raised that question
and gave a detailed discussion about it \cite{Karol1998}. To solve
the problem, as suggested by the definition of separability
probability, we need to calculate two volumes: (1) the volume of the
set of all states acting on the same Hilbert space and (2) the
volume of the set of all separable states acting on the bipartite
Hilbert spaces. Luckily, the volume of the set of all states with
respect to the Hilbert-Schmidt measure was computed by
\.{Z}yczkowski and Sommers \cite{Karol2003} and Andai
\cite{Andai2006}. A review with detailed reasoning for such volume
formula can be found in Zhang \cite{Zhang2015}. Computing the
separability probability of two-qubit quantum states relative to the
Hilbert-Schmidt measure is one of the simplest yet challenging
problems. Numerical simulations lead to intriguing formulas for
separability probability, presented in 2013 by Slater
\cite{Slater2012}: the separability probability for \emph{real}
two-qubit state is
\begin{eqnarray*}
\frac{\vol_{\rH\rS}\Pa{\rD_{\mathrm{sep}}(\real^2\ot\real^2)}}{\vol_{\rH\rS}\Pa{\density{\real^2\ot\real^2}}}=\frac{29}{64}
\end{eqnarray*}
and, for \emph{complex} two-qubit it is
\begin{eqnarray*}
\frac{\vol_{\rH\rS}\Pa{\rD_{\mathrm{sep}}(\complex^2\ot\complex^2)}}{\vol_{\rH\rS}\Pa{\density{\complex^2\ot\complex^2}}}=\frac8{33},
\end{eqnarray*}
where $\density{\mathbb{K}^m\ot\mathbb{K}^n}$ stands for the set of
all bipartite density matrices acting on
$\mathbb{K}^m\ot\mathbb{K}^n$, and
$\rD_{\mathrm{sep}}\Pa{\mathbb{K}^m\ot\mathbb{K}^n}$ is the set of
separable matrices in $\rD\Pa{\mathbb{K}^m\ot\mathbb{K}^n}$. Here
$\mathbb{K}$ equals $\real$ or $\complex$. The \emph{real} case has
been proved by Lovas and Andai \cite{Lovas2016b}. But the
\emph{complex} case is still open at present. There are, however,
results of Aubrun, Szarek, and Ye for the \emph{asymptotic}
separability probability in the limit of the dimension of the state
system tending to infinity \cite{Aubrun2006,Aubrn2014}, which are
very interesting from a mathematical point of view, complementing
the conjectured exact formulas in the low dimensional setting.


The purpose of this article is to infer some information about the
set of bipartite quantum states that are locally diagonalizable. We
have the following new results: a necessary and sufficient condition
for local diagonalizability in the qubit-qubit case (Theorem
\ref{th:loc-diag-cond}) and the Hilbert-Schmidt volume of all
locally diagonalizable states (Theorem
\ref{th:vol-lu-for-two-qubit}). The celebrated Harish-Chandra's
volume formula allows us to calculate the volume of a coadjoint
orbit for a regular point in a specified Weyl chamber. By modifying
Harish-Chandra's formula, we give, for the first time, a specific
formula for the volume of a coadjoint local unitary orbit of a
regular point in a specified Weyl chamber. As an application, we
calculate the volume of the set of all bipartite quantum states that
are locally unitary equivalent to a diagonal quantum state.


Here is the outline of the paper: After introducing basic notions
that we use, we derive, in
Section~\ref{sect:vol-local-diagonal-form}, the joint probability
distribution density of all eigenvalues of all locally
diagonalizable bipartite states. We also present a necessary and
sufficient condition for a two-qubit to be locally unitary (LU)
equivalent to a diagonal form. We then proceed to calculate the
Hilbert-Schmidt (HS) volume of locally diagonalizable states
(Theorem~\ref{th:vol-lu-for-two-qubit}). In
Section~\ref{sect:Harish-Chandra}, we apply Harish-Chandra's volume
formula to specific cases such as the unitary group and the tensor
product of two unitary groups. Two main results of the Section are
Theorems~\ref{th:vol-tensor-group} and \ref{th:Vol-LUO}; the first
yields an analytical formula for the volume of the tensor product of
two unitary groups, and the second leads to the conclusion that the
volume of a locally unitary orbit in the tensor product case equals
the product of the volumes of the factors. We conclude in
Section~\ref{sect:discuz-&-conclusion} with discussion and several
open problems.

\section{Volume of the set of all locally diagonalizable
states}\label{sect:vol-local-diagonal-form}

Suppose we have two finite-dimensional Hilbert spaces $\cH$ and
$\cK$. Specifically, let $\cH=\complex^m$ and $\cK=\complex^n$.
Chose the standard basis $\set{\ket{i}}^m_{i=1}$ and
$\set{\ket{j}}^n_{j=1}$ for $\complex^m$ and $\complex^n$,
respectively. A qudit is represented by a positive semi-definite
matrix of unit trace and we shall identify the two. Denote by
$\density{\complex^k}$ the set of all $k\times k$ density matrices.
Then the set of all bipartite quantum states is
$\density{\complex^m\ot\complex^n}$. Throughout this paper, we do
not distinguish the meaning of a state and a density matrix. A
bipartite state $\rho_{AB}$ is \emph{separable} if it is a
probabilistic mixture of product states $\rho^A_\mu\ot\rho^B_\mu$
where $\rho^A_\mu\in\density{\complex^m}$ and
$\rho^B_\mu\in\density{\complex^n}$. In other words, $\rho_{AB}$ is
separable if
\begin{eqnarray}
\rho_{AB} = \sum_\mu p_\mu\rho^A_\mu\ot\rho^B_\mu,
\end{eqnarray}
where $\set{p_\mu}_\mu$ is a probability distribution. If a bipartite
state is not separable then it is said to be \emph{entangled.}
Thanks to the Spectral Decomposition Theorem, which says that the orbit of a
Hermitian matrix under the adjoint action of unitary matrices
contains a diagonal matrix,
we see that a bipartite state
$\rho_{AB}\in\density{\complex^m\ot\complex^n}$ can be diagonalized
by a global unitary matrix $U\in\rU(mn)$, where $\rU(mn)$ can be
understood as the unitary group
comprising all unitary matrices on $\complex^m\ot\complex^n$;
thus,
\begin{eqnarray}\label{eq:LU-form}
\rho_{AB} = U\Lambda U^\dagger,\quad
\Lambda=\sum^m_{i=1}\sum^n_{j=1}\lambda_{i,j}\out{i,j}{i,j},
\end{eqnarray}
where $\dagger$ denotes the adjoint. Note that all
eigenvalues $\lambda_{ij}$ of $\rho_{AB}$ are indexed by two
indices.


A bipartite state $\rho_{AB}$ is said to be \emph{locally
diagonalizable} if it is diagonalizable in the following manner:
\begin{eqnarray}\label{eq:LD}
\rho_{AB} = (U_A\ot U_B)\Lambda (U_A\ot U_B)^\dagger
\end{eqnarray}
for a simple tensor $U_A\ot U_B\in\rU(m)\ot\rU(n))$; we may, in fact, require that $U_A\ot U_B\in\rS\rU(m)\ot\rS\rU(n)$. We denote
the set of all locally diagonalizable bipartite states from
$\density{\complex^m\ot\complex^n}$ by the following notation:
\begin{eqnarray}
\sD_{\rL\rU}(\complex^m\ot\complex^n) :=
\Set{\rho\in\density{\complex^m\ot\complex^n}: \rho \text{ is
locally diagonalizable}}.
\end{eqnarray}
Locally diagonalizability
is intimately related to what is known as local unitary equivalence; two
bipartite states $\rho$ and $\rho'$  in
$\density{\complex^m\ot\complex^n}$ are said to be \emph{locally
unitary (LU) equivalent} if
\begin{eqnarray}
\rho' = (U\ot V)\rho (U\ot V)^\dagger
\end{eqnarray}
for some simple tensor $U\ot V\in\rS\rU(m)\ot\rS\rU(n)$.
So a bipartite state is locally diagonalizable if and only if
it is LU equivalent to a diagonal state.


Not every bipartite state can be locally
diagonalizable, and even so for separable states. This can be seen
from simple dimension counting \cite{Karol2006}. It is easily
seen that $\dim\Pa{\density{\complex^m\ot\complex^n}}=(mn)^2-1$.
For the submanifold
$\sD_{\rL\rU}(\complex^m\ot\complex^n)$ we have the following
identification:
\begin{eqnarray}
\sD_{\rL\rU}(\complex^m\ot\complex^n) &\simeq&
(\rU(m)\ot\rU(n))/(T_{(m)}\ot T_{(n)})\times \Delta_{mn-1}\\
&\simeq& (\rU(m)/T_{(m)})\ot(\rU(n)/ T_{(n)})\times \Delta_{mn-1},
\end{eqnarray}
where $\Delta_{k-1}:= \Set{(p_1,\ldots,p_k)\in\real^k_+:
\sum^k_{j=1}p_j=1}$ is the $(k-1)$-dimensional probability simplex, and
 $T_{(k)}$ denotes
the (standard) maximal tori of the compact Lie group $\rU(k)$
(more on this in Section~\ref{sect:Harish-Chandra}).
Therefore
\begin{eqnarray}\label{eq:localdim}
\dim\Pa{\sD_{\rL\rU}(\complex^m\ot\complex^n)}=(m^2-1)+(n^2-1)+(m-1)(n-1).
\end{eqnarray}
We note that the minimum of
$\dim\Pa{\sD_{\rL\rU}(\complex^m\ot\complex^n)}$ for  fixed
$d:=mn$ is achieved at $m=n=\sqrt{d}$, while the maximum is achieved
at $m=1$ or $n=1$. We also note that the set of all product mixed
states form an $(m^2+n^2-2)$-dimensional subset of
$\sD_{\rL\rU}(\complex^m\ot\complex^n)$.


Before proceeding further, a few words on the notion of the volume
of a smooth manifold is in order (for details, we refer to
\cite[Sec.~3.13]{Duistermaat2000}). Recall that an $n$-dimensional
oriented manifold $\bsM$ with a pseudo-Riemannian metric $g$ has a
standard volume form $\omega$, known as the \emph{Riemannian volume
form,} whose expression in an oriented chart $(x^1,\dotsc,x^n)$ is
given by
\begin{eqnarray*}
\omega=\sqrt{\det(g)}\dif x^1\wedge\cdots\wedge\dif x^n.
\end{eqnarray*}
If $D$ is a domain of integration in $\bsM$, then
\begin{eqnarray*}
\vol_g(D):=\int_D\omega
\end{eqnarray*}
is called the \emph{Riemannian volume} of $D$.


Of special interest is the case where $\bsM$ is the set of all
non-degenerate full-ranked density matrices from
$\density{\complex^d}$. (It is well-known that degenerate density
matrices in $\density{\complex^d}$ form a subset of zero-measure.)
On $\density{\complex^d}$ we have the \emph{Hilbert-Schmidt} inner
product, which is defined by
\begin{eqnarray*}
\Inner{X}{Y}:=\Tr{X^\dagger Y}.
\end{eqnarray*}
Differentiating this inner product yields a metric on $\bsM$ which
we denote by  $g_{\rH\rS}$. We shall denote by $\vol_{\rH\rS}$ the
Riemannian volume form associated with $g_{\rH\rS}$ and refer to the
volume measured by $\vol_{\rH\rS}$ as the \emph{Hilbert-Schmidt (HS)
volume.} Because the Hilbert-Schmidt inner product is invariant
under the adjoint action, the induced metric and the associated
Riemannian volume form are invariant. So the measure on
$\density{\complex^d}$ induced by $\vol_{\rH\rS}$ is a constant
multiple of the normalized Haar measure.


In order to compute the HS volume of all locally diagonalizable
bipartite states, we wish to parametrize such states. Eigenvalues
can serve that purpose, and knowing the density of eigenvalues
essentially solves the question of finding the HS volume. The
following lemma provides that density:
\begin{lem}\label{lem:joint-prob-dist}
The joint probability density of eigenvalues of all locally
diagonalizable bipartite states in
$\sD_{\rL\rU}(\complex^m\ot\complex^n)$ is given by
\begin{eqnarray}
\bP(\Lambda)\propto
 \biggl[\prod_{1\leqslant i<i'\leqslant
m}\sum^n_{j=1}(\lambda_{ij}-\lambda_{i'j})^2\biggr] \biggl[ \prod_{1\leqslant
j<j'\leqslant
n}\sum^m_{i=1}(\lambda_{ij}-\lambda_{ij'})^2\biggr][\dif\Lambda],
\end{eqnarray}
where $[\dif\Lambda]:=\prod^m_{i=1}\prod^n_{j=1}\dif\lambda_{ij}$ is
the Lebesgue volume element for the diagonal matrix
$\Lambda=\sum^m_{i=1}\sum^n_{j=1}\lambda_{i,j}\out{i,j}{i,j}$.
\end{lem}

\begin{proof}
Let $\rho\in\sD_{\rL\rU}(\complex^m\ot\complex^n)$. Then  $\rho=(U\ot V)\Lambda(U\ot V)^\dagger$ for some
$U\in\rU(m)/T_{(m)}$ and some $V\in\rU(n)/T_{(n)}$. So
\begin{eqnarray}
\dif \rho = (U\ot V)\Pa{\dif\Lambda + \Br{\dif G,\Lambda}}(U\ot
V)^\dagger.
\end{eqnarray}
Here $\dif G = (U\ot V)^\dagger\dif (U\ot V) = \dif
G_1\ot\I_n+\I_m\ot\dif G_2$, where $\dif G_1= U^\dagger\dif U$ and $
\dif G_2= V^\dagger\dif V$. It suffices to identify the volume
element generated by $\Br{\dif G,\Lambda}$. We have that
\begin{eqnarray*}
&&\Inner{\dif \rho}{\dif \rho} = \Inner{\dif\Lambda + \Br{\dif
G,\Lambda}}{{\dif\Lambda + \Br{\dif G,\Lambda}}} \\
&&= \Inner{\dif\Lambda}{\dif\Lambda}+\Inner{\dif\Lambda}{\Br{\dif
G,\Lambda}}+\Inner{\Br{\dif G,\Lambda}}{\dif\Lambda}+
\Inner{\Br{\dif G,\Lambda}}{\Br{\dif G,\Lambda}}\\
&&= \Tr{\dif\Lambda^2}+\Tr{\Br{\dif G,\Lambda}^\dagger\Br{\dif
G,\Lambda}},
\end{eqnarray*}
where
\begin{eqnarray*}
\Tr{\Br{\dif G,\Lambda}^\dagger\Br{\dif G,\Lambda}}
=2\Tr{\Lambda\dif G\Lambda\dif G} -2\Tr{\Lambda^2\dif G^2}.
\end{eqnarray*}
For $\Lambda=\sum^m_{i=1}\sum^n_{j=1}\lambda_{ij}\out{ij}{ij}$, we
 have
\begin{eqnarray*}
&&\Tr{\Br{\dif G,\Lambda}^\dagger\Br{\dif
G,\Lambda}}=2\sum_{i,j}\sum_{i',j'}
\lambda_{ij}\lambda_{i'j'}\Innerm{ij}{\dif
G}{i'j'}\Innerm{i'j'}{\dif G}{ij} -
2\sum_{i,j}\lambda^2_{ij}\Innerm{ij}{\dif G^2}{ij}\\
&&=-2\sum_{i,j}\sum_{i',j'}
\lambda_{ij}\lambda_{i'j'}\abs{\Innerm{ij}{\dif G}{i'j'}}^2 -
2\sum_{i,j}\lambda^2_{ij}\Innerm{ij}{\dif G^2}{ij}\\
&&=2\sum_{i,j}\sum_{i',j'}\lambda^2_{ij}\abs{\Innerm{ij}{\dif
G}{i'j'}}^2 -2\sum_{i,j}\sum_{i',j'}\lambda_{ij}\lambda_{i'j'}
\abs{\Innerm{ij}{\dif
G}{i'j'}}^2\\
&&=2\sum_{i,j}\sum_{i',j'}\lambda_{ij}(\lambda_{ij}-\lambda_{i'j'})\abs{\Innerm{ij}{\dif
G}{i'j'}}^2.
\end{eqnarray*}
Note that
\begin{eqnarray*}
\abs{\Innerm{ij}{\dif G}{i'j'}}^2 &=& \Pa{\Innerm{i}{\dif
G_1}{i'}\iinner{j}{j'}+\Innerm{j}{\dif
G_2}{j'}\iinner{i}{i'}}\Pa{\overline{\Innerm{i}{\dif
G_1}{i'}\iinner{j}{j'}}+\overline{\Innerm{j}{\dif
G_2}{j'}\iinner{i}{i'}}}\\
&=&\abs{\Innerm{i}{\dif G_1}{i'}}^2\delta_{jj'} +
\abs{\Innerm{j}{\dif G_2}{j'}}^2\delta_{ii'}\\
&& + \Innerm{i}{\dif G_1}{i'}\overline{\Innerm{j}{\dif
G_2}{j'}}\delta_{ii'}\delta_{jj'}+\overline{\Innerm{i}{\dif
G_1}{i'}}\Innerm{j}{\dif G_2}{j'}\delta_{ii'}\delta_{jj'}.
\end{eqnarray*}
Thus,
\begin{eqnarray*}
\Tr{\Br{\dif G,\Lambda}^\dagger\Br{\dif
G,\Lambda}}&=&2\sum_{i,j}\sum_{i',j'}\lambda_{ij}(\lambda_{ij}-\lambda_{i'j'})\abs{\Innerm{i}{\dif G_1}{i'}}^2\delta_{jj'}\\
&&+2\sum_{i,j}\sum_{i',j'}\lambda_{ij}(\lambda_{ij}-\lambda_{i'j'})\abs{\Innerm{j}{\dif G_2}{j'}}^2\delta_{ii'}\\
&&+2\sum_{i,j}\sum_{i',j'}\lambda_{ij}(\lambda_{ij}-\lambda_{i'j'})\Innerm{i}{\dif
G_1}{i'}\overline{\Innerm{j}{\dif
G_2}{j'}}\delta_{ii'}\delta_{jj'}\\
&&+2\sum_{i,j}\sum_{i',j'}\lambda_{ij}(\lambda_{ij}-\lambda_{i'j'})\overline{\Innerm{i}{\dif
G_1}{i'}}\Innerm{j}{\dif G_2}{j'}\delta_{ii'}\delta_{jj'}.
\end{eqnarray*}
That is,
\begin{eqnarray*}
\Tr{\Br{\dif G,\Lambda}^\dagger\Br{\dif
G,\Lambda}}&=&2\sum_{i',i,j}\lambda_{ij}(\lambda_{ij}-\lambda_{i'j})\abs{\Innerm{i}{\dif
G_1}{i'}}^2\\
&&+2\sum_{i,j,j'}\lambda_{ij}(\lambda_{ij}-\lambda_{ij'})\abs{\Innerm{j}{\dif
G_2}{j'}}^2.
\end{eqnarray*}
Since $\abs{\Innerm{i}{\dif G_1}{i'}}^2 = \abs{\Innerm{i'}{\dif
G_1}{i}}^2$ (both vanish if $i=i'$ because $\dif G_1$ is
skew-Hermitian), we have
\begin{eqnarray*}
&&\sum_{i'\neq
i}\Br{\sum_j\lambda_{ij}(\lambda_{ij}-\lambda_{i'j})}\abs{\Innerm{i}{\dif
G_1}{i'}}^2 \\
&&= \sum_{i<
i'}\Br{\sum_j\lambda_{ij}(\lambda_{ij}-\lambda_{i'j})}\abs{\Innerm{i}{\dif
G_1}{i'}}^2 + \sum_{i>
i'}\Br{\sum_j\lambda_{ij}(\lambda_{ij}-\lambda_{i'j})}\abs{\Innerm{i}{\dif
G_1}{i'}}^2\\
&&=\sum_{i< i'}\Br{\sum_j\lambda_{ij}(\lambda_{ij}-\lambda_{i'j})+
\sum_j\lambda_{i'j}(\lambda_{i'j}-\lambda_{ij})}\abs{\Innerm{i}{\dif
G_1}{i'}}^2\\
&&=\sum_{i<
i'}\Br{\sum_j(\lambda_{ij}-\lambda_{i'j})^2}\abs{\Innerm{i}{\dif
G_1}{i'}}^2.
\end{eqnarray*}
Therefore,
\begin{eqnarray*}
&&\Inner{\dif \rho}{\dif \rho}
=\sum_{i,j}\dif\lambda^2_{ij}+2\sum_{i'\neq
i,j}\lambda_{ij}(\lambda_{ij}-\lambda_{i'j})\abs{\Innerm{i}{\dif
G_1}{i'}}^2+2\sum_{i,j\neq
j'}\lambda_{ij}(\lambda_{ij}-\lambda_{ij'})\abs{\Innerm{j}{\dif
G_2}{j'}}^2\\
&&=\sum_{i,j}\dif\lambda^2_{ij}+2\sum_{i<i'}\Pa{\sum_j(\lambda_{ij}-\lambda_{i'j})^2}\abs{\Innerm{i}{\dif
G_1}{i'}}^2+2\sum_{j<j'}\Pa{\sum_i(\lambda_{ij}-\lambda_{ij'})^2}\abs{\Innerm{j}{\dif
G_2}{j'}}^2.
\end{eqnarray*}
This shows that the Hilbert-Schmidt volume element is given by
\begin{eqnarray*}
&&[\dif
\rho]=\prod_{i<i'}2\Br{\sum_j(\lambda_{ij}-\lambda_{i'j})^2}\prod_{j<j'}2\Br{\sum_i(\lambda_{ij}-\lambda_{ij'})^2}[\dif\Lambda][\dif
G_1][\dif G_2]\\
&&=2^{\binom{m}{2}+\binom{n}{2}}\Br{\prod_{i<i'}\sum_j(\lambda_{ij}-\lambda_{i'j})^2}\Br{\prod_{j<j'}\sum_i(\lambda_{ij}-\lambda_{ij'})^2}[\dif\Lambda][\dif
G_1][\dif G_2].
\end{eqnarray*}

The measure induced by the Lebesgue volume element $[\dif G_1]$ on
the flag manifold $\rU(m)/T_{(m)}$ is the quotient measure. Thus,
\begin{eqnarray*}
[\dif G_1] = \frac{\vol_{\rH\rS}(\rU(m))}{\vol_{\rH\rS}(T_{(m)})}
\dif \mu_{\mathrm{Haar}}(U),
\end{eqnarray*}
where $\dif \mu_{\mathrm{Haar}}$ denotes the normalized Haar
measure. Similarly for $[\dif G_2]$. Hence,
\begin{eqnarray}\label{eq:[drho]}
[\dif
\rho]&=&
C_{m,n}
\biggl[\prod_{i<i'}\sum_j(\lambda_{ij}-\lambda_{i'j})^2\biggr]
\biggl[\prod_{j<j'}\sum_i(\lambda_{ij}-\lambda_{ij'})^2\biggr][\dif\Lambda]\dif\mu_{\mathrm{Haar}}(U)\dif\mu_{\mathrm{Haar}}(V),
\end{eqnarray}
where
\begin{eqnarray}\label{eq:cmn}
C_{m,n}=2^{\binom{m}{2}+\binom{n}{2}}
\frac{\vol_{\rH\rS}(\rU(m))}{\vol_{\rH\rS}(T_{(m)})}
\frac{\vol_{\rH\rS}(\rU(n))}{\vol_{\rH\rS}(T_{(n)})},
\end{eqnarray}
Integrating over $U$ and $V$ gives the claimed result.

Finally, thanks to Harish-Chandra's volume formula (see Proposition
\ref{prop:hc-vol}), the constant $C_{m,n}$ is completely determined
by the Lie-algebraic properties of $\rU(m)$ and $\rU(n)$.
\end{proof}

\subsection{Necessary and sufficient conditions for locally diagonalizable two-qubits}

For the most part in this paper, we focus on two-qubits.
But first, recall the notion of the Bloch sphere
representation for a single qubit. In quantum mechanics, the Bloch sphere is a
geometrical representation of the pure state space of a two-level
quantum mechanical system (qubit). Any qubit state can be
represented using the Pauli matrices:
\begin{eqnarray}
\sigma_x =
             \begin{pmatrix}
               0 & 1 \\
               1 & 0 \\
             \end{pmatrix},
           \quad \sigma_y =
             \begin{pmatrix}
               0 & -\mathrm{i} \\
               \mathrm{i} & 0 \\
             \end{pmatrix},
           \quad\sigma_z =
                                 \begin{pmatrix}
                                   1 & 0 \\
                                   0 & -1 \\
                                 \end{pmatrix}.
\end{eqnarray}
More precisely, if we write $\boldsymbol{\sigma}:=(\sigma_x,\sigma_y,\sigma_z)$,
then, for any qubit state $\rho$, we have
\begin{eqnarray}\label{eq:generic-qubit}
\rho = \frac12 (\I + \bsr(\rho)\cdot \boldsymbol{\sigma}),
\end{eqnarray}
for some suitable $\bsr(\rho)\in \mathbb{R}^3$, known as the Bloch
vector of $\rho$, satisfying
$\norm{\bsr(\rho)}:=\sqrt{r^2_x+r^2_y+r^2_z}\leqslant1$. The last
term in \eqref{eq:generic-qubit} is the usual `dot product' of
$3$-tuples.


For reference sake, the commutation and anti-commutation relations
satisfied by the Pauli matrices are:
\begin{eqnarray}
\begin{aligned}\relax
[\sigma_a,\sigma_b] &:=\sigma_a\alpha_b-\sigma_b\sigma_a= 2\epsilon_{abc}\mathrm{i}\sigma_c,\\
\set{\sigma_a,\sigma_b} &:= \sigma_a\sigma_b + \sigma_b\sigma_a
= 2\delta_{ab}\I,
\end{aligned}
\end{eqnarray}
where $\epsilon_{ijk}$ is the Levi-Civita symbol, and $\delta_{ij}$ is the Kronecker delta. The equivariance relation satisfied
by the Bloch vector is as follows:
For $U\in\rS\rU(2)$,
there is some $O\in\rS\rO(3)$ such that
\begin{eqnarray}
\bsr(U\rho U^\dagger) = O\bsr(\rho).
\end{eqnarray}
In other words, the adjoint action of a unitary matrix
$U\in\rS\rU(2)$ on a qubit state $\rho$ amounts to  a rotation of the
corresponding Bloch vector of $\rho$.
We point out that the correspondence
\begin{equation}
\begin{array}{ccc}
\rS\rU(2)&\to&\rS\rO(3)\\
U & \mapsto &O
\end{array}\label{eq:su2dcvr}
\end{equation}
is the universal double covering for $\rS\rO(3)$.


With the Bloch sphere representation \eqref{eq:generic-qubit} for a
qubit, we have that any two-qubit state can be written in the
following way:
\begin{eqnarray}\label{eq:generic}
\rho_{AB} = \frac14\biggl(\I_2\ot\I_2 +
\bsr\cdot\boldsymbol\sigma\ot\I_2 + \I_2\ot\bss\cdot
\boldsymbol{\sigma} + \sum_{i,j=x,y,z}
t_{ij}\sigma_i\ot\sigma_j\biggr).
\end{eqnarray}
Thus, any two-qubit is given by specifying the $3$-dimensional
vectors $\bsr$ and $\bss$ in $\real^3$ and the real $3\times 3$
matrix $T=(t_{ij})$.


If two-qubit states $\rho_{AB}$ and $\rho'_{AB}$ are LU equivalent,
that is, $\rho'_{AB} = (U_A\ot U_B)\rho_{AB}(U_A\ot U_B)^\dagger$
for some $U_A$ and $U_B$ in $\rS\rU(2)$, then one can directly check
that
 there are
$O_A$ and $O_B$ in $\rS\rO(3)$ such that
\begin{eqnarray}
\begin{gathered}
\bsr' = O_A\bsr,\quad \bss' =
O_B\bss,\\
T' = O_ATO^\t_B.
\end{gathered}
\label{eq:2qlueqv}
\end{eqnarray}
Conversely, the existence of such $O_A$ and $O_B$ implies that
$\rho_{AB}$ and $\rho_{AB}'$ are LU equivalent, thanks to the map
\eqref{eq:su2dcvr} being a covering. The following theorem from
\cite{Jing2015} gives another way to check the condition
\eqref{eq:2qlueqv}:

\begin{thrm}\label{th:Jing}
Two generic two-qubit states are LU equivalent if and only if they
have the same values for the following twelve invariants: For $k=0,1,2$,
\begin{eqnarray}
&&\Innerm{\bsr}{\Pa{TT^\t}^k}{\bsr},\quad
\Innerm{\bss}{\Pa{T^\t T}^k}{\bss},\\
&&\Innerm{\bsr}{\Pa{TT^\t}^k T}{\bss},\quad\Tr{\Br{TT^\t}^{k+1}}.
\end{eqnarray}
\end{thrm}

Our goal of this section is to find an equivalent condition for
local diagonalizability of two-qubits. Although it is not absolutely
necessary, we will present our argument using the notion of
\emph{$X$-states,} that is, states whose density matrices are of the
form
\begin{eqnarray}\label{eq:X}
\rho_X := \begin{pmatrix}
                \rho_{11} & 0 & 0 & \rho_{14} \\
                0 & \rho_{22} & \rho_{23} & 0 \\
                0 & \rho_{32} & \rho_{33} & 0 \\
                \rho_{41} & 0 & 0 & \rho_{44}
              \end{pmatrix}.
\end{eqnarray}
In particular, $\rho_X$ satisfies the following unit trace and positivity conditions:
\begin{enumerate}[(i)]
\item $\sum^4_{j=1}\rho_{jj}=1$
\item $\rho_{22}\rho_{33}\geqslant\abs{\rho_{23}}^2$ and $\rho_{11}\rho_{44}\geqslant\abs{\rho_{14}}^2$
\end{enumerate}
Diagonal states are special cases of $X$-states. We denote by
$\sD_X(\complex^2\ot\complex^2)$ the set of all two-qubit
$X$-states; it is a $7$-dimensional submanifold of
$\density{\complex^2\ot\complex^2}$, while
$\dim(\density{\complex^2\ot\complex^2})=15$. Thus, owing to
Equation \eqref{eq:localdim}, we have
$\dim(\sD_X(\complex^2\ot\complex^2))=\dim(\sD_{\rL\rU}(\complex^2\ot\complex^2))=7$.
The Hilbert-Schmidt volume of $\sD_X(\complex^2\ot\complex^2)$ has
been calculated\footnote{The Hilbert-Schmidt volume stated here is a
correction to the result of Milz and Strunz \cite{Milz2014}; they
misused a factor leading to an incorrect Hilbert-Schmidt volume of
${\pi^2}/{5040}$. Their calculation of the Euclid volume of
$X$-states is correct.} in \cite{Milz2014}:
$\vol_{\rH\rS}\Pa{\sD_X(\complex^2\ot\complex^2)}={\pi^2}/{630}$.

\begin{prop}\label{prop:x-bloch}
If a two-qubit state is an $X$-state $\rho_X$, then it can be
written as
\begin{eqnarray}
\rho_{X} &=& \frac14(\I\ot\I + a_z\sigma_z\ot\I+
b_z\I\ot\sigma_z+r_{xx}\sigma_x\ot\sigma_x\notag\\
&&~~~~+r_{xy}\sigma_x\ot\sigma_y+r_{yx}\sigma_y\ot\sigma_x+r_{yy}\sigma_y\ot\sigma_y+r_{zz}\sigma_z\ot\sigma_z),
\end{eqnarray}
where
\begin{eqnarray*}
\begin{cases}
a_z &= \rho_{11}-\rho_{22} + \rho_{33}-\rho_{44},\\
b_z &= \rho_{11}+\rho_{22} - \rho_{33}-\rho_{44},\\
r_{zz} &= \rho_{11}-\rho_{22} - \rho_{33}+\rho_{44},
\end{cases}
\quad\text{and}\quad
\begin{cases}
r_{xx} &= \rho_{14}+\rho_{23}+\rho_{32}+\rho_{41},\\
r_{xy} &= \mathrm{i}(\rho_{14}+\rho_{23}-\rho_{32}-\rho_{41}),\\
r_{yx} &= \mathrm{i}(\rho_{14}-\rho_{23}+\rho_{32}-\rho_{41}),\\
r_{yy} &= -\rho_{14}+\rho_{23}+\rho_{32}-\rho_{41}.
\end{cases}
\end{eqnarray*}
Moreover, all eigenvalues of $\rho_X$ are given by
\begin{eqnarray}
\lambda_{1,2}(\rho_X) &=&
\frac14\Pa{(1+r_{zz})\pm\sqrt{(a_z+b_z)^2+(r_{xx}-r_{yy})^2 +
(r_{xy}+r_{yx})^2}},\\
\lambda_{3,4}(\rho_X) &=&
\frac14\Pa{(1-r_{zz})\pm\sqrt{(a_z-b_z)^2+(r_{xx}+r_{yy})^2 +
(r_{xy}-r_{yx})^2}}.
\end{eqnarray}
\end{prop}

Applying Proposition \ref{prop:x-bloch} to a $4\times 4$ diagonal
state $\Lambda=\diag(\lambda_1,\lambda_2,\lambda_3,\lambda_4)$, we
have that
\begin{eqnarray}
\Lambda = \frac14(\I_2\ot\I_2 + a_z\sigma_z\ot\I_2+
b_z\I_2\ot\sigma_z+r_{zz}\sigma_z\ot\sigma_z),
\end{eqnarray}
where
\begin{eqnarray}\label{eq:triple}
\begin{cases}
a_z = \lambda_1-\lambda_2 + \lambda_3-\lambda_4\\
b_z = \lambda_1+\lambda_2 - \lambda_3-\lambda_4\\
r_{zz} = \lambda_1-\lambda_2 - \lambda_3+\lambda_4
\end{cases}.
\end{eqnarray}
Then $\bsr(\Lambda)=a_z\ket{3}, \bss(\Lambda)=b_z\ket{3}$, and
$T(\Lambda) = r_{zz}\out{3}{3}$, where $\ket{3}:=(0,0,1)^\t$.
Applying Theorem~\ref{th:Jing} to a state $\rho$ that is LU
equivalent to the diagonal form $\Lambda$ gives us the following
result:

\begin{thrm}\label{th:loc-diag-cond}
A generic two-qubit state $\rho$ is locally diagonalizable if and
only if the following conditions are satisfied:
\begin{eqnarray}
\bsr = a_zO_A\ket{3},\quad\bss=b_zO_B\ket{3},\quad
T=r_{zz}O_A\out{3}{3}O^\t_B,
\end{eqnarray}
where $O_A$ and $O_B$ are in $\rS\rO(3)$, and the triple $(a_z,b_z,r_{zz})$ is
given by \eqref{eq:triple}.
\end{thrm}
The above theorem characterizes locally diagonalizable two-qubit
states. Its generalization to higher dimensions is apparently
unknown at present.

\subsection{The Hilbert-Schmidt volume of locally diagonalizable two-qubits}

The primary goal of this section is to compute the HS volume of
locally diagonalizable two-qubits. The result is the following
theorem:

\begin{thrm}\label{th:vol-lu-for-two-qubit}
The Hilbert-Schmidt volume of
$\sD_{\rL\rU}(\complex^2\ot\complex^2)$ is given by
\begin{eqnarray}
\vol_{\rH\rS}\bigl(\sD_{\rL\rU}(\complex^2\ot\complex^2)\bigr)=
\frac{(4\pi)^2}{105}.
\end{eqnarray}
\end{thrm}

\begin{proof}
If $\rho\in
\sD_{\rL\rU}(\complex^2\ot\complex^2)$, then
\begin{eqnarray}
\rho = (U\ot V)\Lambda (U\ot V)^\dagger,
\end{eqnarray}
where $U$ and $V$ are in $\rU(2)$, and
$\Lambda=\diag(\lambda_1,\lambda_2,\lambda_3,\lambda_4)$ with
$\lambda_j$'s being pairwise different and satisfying
$\sum_j\lambda_j=1$. In the proof of
Lemma~\ref{lem:joint-prob-dist}, we see from \eqref{eq:[drho]} that
\begin{eqnarray*}
[\dif\rho] =
C_{2,2}\Br{(\lambda_1-\lambda_2)^2+(\lambda_3-\lambda_4)^2}
\Br{(\lambda_1-\lambda_3)^2+(\lambda_2-\lambda_4)^2}[\dif\Lambda]\dif
\mu_{\mathrm{Haar}}(U)\dif \mu_{\mathrm{Haar}}(V).
\end{eqnarray*}
Then,
\begin{eqnarray}
\delta(1-\Tr{\rho})[\dif\rho] &=&
C_{2,2}\delta\Pa{1-\sum_j\lambda_j} \notag \\
&&\times\Br{(\lambda_1-\lambda_2)^2+(\lambda_3-\lambda_4)^2}\Br{(\lambda_1-\lambda_3)^2+(\lambda_2-\lambda_4)^2}[\dif\Lambda].
\end{eqnarray}
Thus
\begin{eqnarray*}
&&\int_{\sD_{\rL\rU}(\complex^2\ot\complex^2)}\delta(1-\Tr{\rho})[\dif\rho]\\
&&= C_{2,2}
\int\delta\Pa{1-\sum_j\lambda_j}\Br{(\lambda_1-\lambda_2)^2+(\lambda_3-\lambda_4)^2}
\Br{(\lambda_1-\lambda_3)^2+(\lambda_2-\lambda_4)^2}[\dif\Lambda].
\end{eqnarray*}
It is known that (for details, see \cite{Zhang2015})
\begin{eqnarray}
\vol_{\rH\rS}(\rU(k)) =
\frac{(2\pi)^{\binom{k+1}{2}}}{\prod^k_{j=1}\Gamma(j)}.
\end{eqnarray}
So $\vol_{\rH\rS}(\rU(2))=(2\pi)^3$ and $\vol_{\rH\rS}(T_{(2)}) =
\vol_{\rH\rS}(\rU(1)^2) = (\vol_{\rH\rS}(\rU(1))^2 = (2\pi)^2$.
Therefore,
\begin{eqnarray*}
C_{2,2} = (4\pi)^2
\end{eqnarray*}
and we have
\begin{eqnarray}
&&\vol_{\rH\rS}\Pa{\sD_{\rL\rU}(\complex^2\ot\complex^2)}=\int_{\sD_{\rL\rU}(\complex^2\ot\complex^2)}\delta(1-\Tr{\rho})[\dif\rho]\\
&&= (4\pi)^2\int
\delta\Pa{1-\sum_j\lambda_j}\Br{(\lambda_1-\lambda_2)^2+(\lambda_3-\lambda_4)^2}\Br{(\lambda_1-\lambda_3)^2+(\lambda_2-\lambda_4)^2}\prod^4_{j=1}\dif\lambda_j.
\end{eqnarray}

It remains to evaluate the last integral. Note that
\begin{eqnarray*}
&&\Br{(\lambda_1-\lambda_2)^2+(\lambda_3-\lambda_4)^2}\Br{(\lambda_1-\lambda_3)^2+(\lambda_2-\lambda_4)^2}
\\&&=(\lambda_1-\lambda_2)^2(\lambda_1-\lambda_3)^2
 +(\lambda_3-\lambda_1)^2(\lambda_3-\lambda_4)^2 + (\lambda_2-\lambda_1)^2(\lambda_2-\lambda_4)^2
+(\lambda_4-\lambda_2)^2(\lambda_4-\lambda_3)^2.
\end{eqnarray*}
The following four integrals are equal:
\begin{eqnarray*}
\int
\delta\Pa{1-\sum_j\lambda_j}(\lambda_1-\lambda_2)^2(\lambda_1-\lambda_3)^2\prod^4_{j=1}\dif\lambda_j,\\
\int
\delta\Pa{1-\sum_j\lambda_j}(\lambda_3-\lambda_1)^2(\lambda_3-\lambda_4)^2\prod^4_{j=1}\dif\lambda_j,\\
\int
\delta\Pa{1-\sum_j\lambda_j}(\lambda_2-\lambda_1)^2(\lambda_2-\lambda_4)^2\prod^4_{j=1}\dif\lambda_j,\\
\int
\delta\Pa{1-\sum_j\lambda_j}(\lambda_4-\lambda_2)^2(\lambda_4-\lambda_3)^2\prod^4_{j=1}\dif\lambda_j.\\
\end{eqnarray*}
Let
\begin{eqnarray}\label{eq:f(t)}
f(t) =\int
\delta\Pa{t-\sum_j\lambda_j}(\lambda_1-\lambda_2)^2(\lambda_1-\lambda_3)^2\prod^4_{j=1}\dif\lambda_j.
\end{eqnarray}
Performing the Laplace transformation $(t\to s)$ on $f(t)$, we get,
for $s>0$,
\begin{eqnarray}
\widetilde f(s) &=& \sL(f)(s) =\int^\infty_0 f(t)e^{-st}\dif t\\
&=&\int^\infty_0\int^\infty_0\int^\infty_0\int^\infty_0
\Br{\int^\infty_0\dif
te^{-st}\delta\Pa{t-\sum_j\lambda_j}}(\lambda_1-\lambda_2)^2(\lambda_1-\lambda_3)^2\prod^4_{j=1}\dif\lambda_j.
\\
&=&\int^\infty_0\int^\infty_0\int^\infty_0\int^\infty_0
\exp\Pa{-s\sum_j\lambda_j}(\lambda_1-\lambda_2)^2(\lambda_1-\lambda_3)^2\prod^4_{j=1}\dif\lambda_j.
\end{eqnarray}
By change of variables,
\begin{eqnarray}
\widetilde f(s) &=& s^{-8}
\int^\infty_0\int^\infty_0\int^\infty_0\int^\infty_0
\exp\Pa{-\sum^4_{j=1}x_j}(x_1-x_2)^2(x_1-x_3)^2\prod^4_{j=1}\dif
x_j\\
&=&s^{-8} \int^\infty_0\int^\infty_0\int^\infty_0
\exp\Pa{-\sum^3_{j=1}x_j}(x_1-x_2)^2(x_1-x_3)^2\prod^3_{j=1}\dif
x_j.
\end{eqnarray}
Since
\begin{eqnarray*}
(x_1-x_2)^2(x_1-x_3)^2 = x^4_1 - 2(x_2+x_3)x^3_1 +
(x^2_2+x^2_3+4x_2x_3)x^2_1 - 2x_2x_3(x_2+x_3)x_1 + x^2_2x^2_3,
\end{eqnarray*}
we have
\begin{eqnarray*}
&&\int^\infty_0\dif x_1 e^{-x_1}(x_1-x_2)^2(x_1-x_3)^2 \\
&&= \Gamma(5) - 2(x_2+x_3)\Gamma(4) + (x^2_2+x^2_3+4x_2x_3)\Gamma(3)
- 2x_2x_3(x_2+x_3)\Gamma(2) + x^2_2x^2_3,
\end{eqnarray*}
where we utilized the integral representation of the Gamma function
$\Gamma(z)=\int^\infty_0 x^{z-1}e^{-x}\dif x$. Hence,
\begin{eqnarray*}
&&\int^\infty_0\int^\infty_0\int^\infty_0\dif x_1\dif x_2\dif x_3 e^{-x_1}e^{-x_2}e^{-x_3}(x_1-x_2)^2(x_1-x_3)^2 \\
&&= \Gamma(5) - 2[\Gamma(2)+\Gamma(2)]\Gamma(4) +
[\Gamma(3)+\Gamma(3)+4\Gamma(2)\Gamma(2)]\Gamma(3) \\
&&~~~~- 2[\Gamma(3)\Gamma(2)+\Gamma(2)\Gamma(3)]\Gamma(2) +
\Gamma(3)\Gamma(3)\\
&&=12.
\end{eqnarray*}
Therefore,
\begin{eqnarray}
\widetilde f(s) = 12\cdot s^{-8}.
\end{eqnarray}
Then,
\begin{eqnarray}
f(t) = \sL^{-1}(\widetilde f)(t) = 12\cdot \frac{t^7}{7!} =
\frac1{420}t^7.
\end{eqnarray}
Finally,
\begin{eqnarray}
\vol_{\rH\rS}\bigl(\sD_{\rL\rU}(\complex^2\ot\complex^2)\bigr) =
(4\pi)^2\cdot4f(1) = \frac{(4\pi)^2}{105}.
\end{eqnarray}
This completes the proof.
\end{proof}

\begin{remark}\label{rem:dirichlet}
The evaluation of $f(1)$ from \eqref{eq:f(t)} could have been done
using the formulas for the density of Dirichlet distributions; the
Dirichlet distribution of order $N$ with parameters $\alpha_j>0$
($j=1,\ldots,N$) has a probability density function relative to the
Lebesgue measure on $\real^{N-1}$ given by
\begin{eqnarray*}
p(x_1,\ldots,x_N;
\alpha_1,\ldots,\alpha_N):=C(\alpha_1,\ldots,\alpha_N)
\delta\Pa{1-\sum^N_{k=1}x_k}\prod^N_{k=1}x^{\alpha_k-1}_k,
\end{eqnarray*}
where the normalization constant $C(\alpha_1,\ldots,\alpha_N)$ is given by
\begin{eqnarray*}
C(\alpha_1,\ldots,\alpha_N) =
\frac{\Gamma(\sum^N_{k=1}\alpha_k)}{\prod^N_{k=1}\Gamma(\alpha_k)}.
\end{eqnarray*}
Here $\Gamma$ denotes the Gamma function as usual.
\end{remark}

\begin{remark}
Seeking to generalize the volume formula, we can attempt to apply
the argument used in the proof of
Theorem~\ref{th:vol-lu-for-two-qubit} to arbitrary bipartite states.
That would lead us to:
\begin{eqnarray}
&&\vol_{\rH\rS}\bigl(\sD_{\rL\rU}(\complex^m\ot\complex^n)\bigr) =
\int_{\sD_{\rL\rU}(\complex^m\ot\complex^n)}\delta(1-\Tr{\rho})[\dif\rho]\\
&&= C_{m,n}
\int\delta\biggl(1-\sum^m_{i=1}\sum^n_{j=1}\lambda_{ij}\biggr)
\notag\\
&&~~~~\times \biggl(\prod_{1\leqslant i<i'\leqslant m}
\sum^n_{j=1}(\lambda_{ij}-\lambda_{i'j})^2\biggr)
\biggl(\prod_{1\leqslant j<j'\leqslant n}
\sum^m_{i=1}(\lambda_{ij}-\lambda_{ij'})^2\biggr)
\prod^m_{i=1}\prod^n_{j=1}\dif\lambda_{ij}.\label{eq:expansion-local}
\end{eqnarray}
Here the constant $C_{m,n}$ is from \eqref{eq:cmn}. The above
integral can in principle be evaluated for a given pair $(m,n)$,
giving us a volume formula. But carrying out the computation seems
complicated (we tried to use computers for the qubit-qutrit case
$(m,n)=(2,3)$ without success). We certainly do not have a unified
closed expression for the integral \eqref{eq:expansion-local}. We
can still try to obtain some insight. Let
\begin{eqnarray*}
P(\lambda_{ij}):= \biggl( \prod_{1\leqslant i<i'\leqslant
m}\sum^n_{j=1} (\lambda_{ij}-\lambda_{i'j})^2\biggr)
\biggl(\prod_{1\leqslant j<j'\leqslant n}\sum^m_{i=1}
(\lambda_{ij}-\lambda_{ij'})^2\biggr).
\end{eqnarray*}
This is a homogeneous multivariate polynomial of $mn$ variables
$\lambda_{ij}$ with integer coefficients. Upon expansion, we get
\begin{eqnarray}
P(\lambda_{ij}) = \sum \prod_{i,j}\lambda^{\gamma_{ij}-1}_{ij},
\end{eqnarray}
where $\gamma_{ij}$'s are positive integers, and the summation is
finite. From Remark~\ref{rem:dirichlet}, we can infer that the HS
volume of all locally diagonalizable bipartite states is always a
power of $\pi$ times a rational number.
\end{remark}

\section{Harish-Chandra's volume formula}\label{sect:Harish-Chandra}

Let $\rU(m)$ be the unitary group acting on $\complex^m$ with Lie
algebra $\u(m)$. Denote by $T_{(m)}$ the standard maximal torus of $\rU(m)$,
namely, the set of diagonal matrices in $\rU(m)$.
Note that $T_{(n)}\cong \rU(1)^{\times n}$ and that the Lie algebra $\liet_{(m)}$ of $T_{(m)}$ is isomorphic to
$\sqrt{-1}\real^m$. Without
loss of generality, we take $\liet_{(m)}$ as the set of diagonal
matrices with purely imaginary diagonal entries. Let
$K=\rU(m)\ot\rU(n)$. Then for the Lie algebra $\k$ of $K$
we have
\begin{eqnarray}
\k = \u(m)\ot\I_n+\I_m\ot\u(n).
\end{eqnarray}
Similarly, for the Lie algebra $\liet$ of the maximal torus $T=T_{(m)}\ot T_{(n)}$
of $K$, we have
\begin{eqnarray}
\liet = \liet_{(m)}\ot\I_n+\I_m\ot\liet_{(n)}.
\end{eqnarray}
It is a routine exercise to see that
$\dim(\liet)=\dim(\liet_{(m)})+\dim(\liet_{(n)})-1=m+n-1$.  We shall see that (Proposition~\ref{prop:tsr-rtsp}) $\dim(\k)=\dim(\u(m))+\dim(\u(n))-1=m^2+n^2-1$. We also
have
\begin{eqnarray}
[\rU(m)\ot\rU(n)]/[T_{(m)}\ot T_{(n)}]\simeq
[\rU(m)/T_{(m)}]\ot[\rU(n)/T_{(n)}].
\end{eqnarray}
It goes without saying that, for volumes of quotient spaces, we
shall use quotient measures. For details on quotient measures, we
refer to \cite[Sec.~3.13]{Duistermaat2000}.

\begin{prop}[Harish-Chandra's volume formula \cite{harishchandra1}]
\label{prop:hc-vol} Let $K$ be a compact, connected Lie group. Let
$T$ be the maximal torus of $K$. Endow $K$ with the metric $g$
induced by an invariant inner product on the Lie algebra
$\mathfrak{k}$ of $K$; endow $T$ with the subspace metric. Then the
Riemannian volumes of $K$ and $T$ satisfy
\begin{eqnarray}
\frac{\vol_g(K)}{\vol_g(T)} =
\prod_{\alpha\in\Phi^+_\k}\frac{2\pi}{\Inner{\alpha}{\varpi}},
\label{eq:hcvf}
\end{eqnarray}
where $\Phi^+_\k$ is the set of all positive roots for $\k$ and
$\varpi:=\frac12\sum_{\alpha\in\Phi^+_\k}\alpha$, which is Weyl
vector, i.e., the half the sum of all positive roots of $\k$.
\end{prop}

\begin{remark}
If we endow the flag manifold $K/T$ with the quotient measure, then
the volume ratio \eqref{eq:hcvf} equals the volume of $K/T$. The
beauty of Harish-Chandra's formula is that the volume ratio is
completely determined by the Lie-algebraic properties of $\k$ and is
independent of the choice of the invariant metric for $K$ and $T$.
For modern expositions on Harish-Chandra's formula, we refer to
\cite[Cor.~7.27]{Berline1992} or
\cite[Eq.~3.14.13]{Duistermaat2000}.
\end{remark}

\subsection{Hilbert-Schmidt volume}

Our focus is on the Hilbert-Schmidt measure on $K$ and $T$. We shall
denote the volume of $K/T$ relative to the quotient measure as
$\vol_{\rH\rS}(K/T)$. Hence,
\begin{eqnarray*}
\vol_{\rH\rS}(K/T)= \frac{\vol_{\rH\rS}(K)}{\vol_{\rH\rS}(T)}.
\end{eqnarray*}

\begin{exam}\label{ex:unvol}
Let us calculate the volume of the flag
manifold $\rU(n)/T_{(n)}$. The Lie algebra $\liet_{(n)}$
is the set of all diagonal matrices in $\u(n)$.
We can take the following set as the set of all positive roots of
$\u(n)$:
$$
\Phi^+_{\u(n)}=\Set{\alpha_{ij}\in\liet^*_{(n)}:\alpha_{ij}(X)=x_i-x_j
\text{ for any }X=\diag(x_1,\ldots,x_n), \ i<j}.
$$
We can view a diagonal matrix $X$ as a real vector $X=(x_1,\ldots,x_n)$. In turn, we may view $\alpha_{ij}$ as the real vector
$$
\alpha_{ij}=(\cdots \overbrace{1}^{i}\cdots
\overbrace{-1}^{j}\cdots)\quad (1\leqslant i<j\leqslant n),
$$
where $\cdot$ stands for zeroes. Then $\alpha_{ij}(X)=\Inner{\alpha_{ij}}{X}$,
where the right-hand side denotes
the dot product of $\alpha_{ij}$ and $X$. The Weyl vector is
$$
\varpi=\frac12\sum_{i<j}\alpha_{ij}=\frac12(n-1,n-3,\ldots,3-n,1-n).
$$
Let $\set{e_j:j=1,\ldots,n}$ be the standard orthonormal basis for $\real^n$.
Then $\alpha_{ij}=e_i - e_j$ and $\varpi=\sum_j\Pa{\frac{n+1}2-j}
e_j$. So $\Inner{\alpha_{ij}}{\varpi} = \varpi_i-\varpi_j =
\Pa{\frac{n+1}2-i}-\Pa{\frac{n+1}2-j}=j-i$.
 Then, by \eqref{eq:hcvf},
\begin{eqnarray}
\vol_{\rH\rS}(\rU(n)/T_{(n)})
= \frac{\vol_{\rH\rS}(\rU(n))}{\vol_{\rH\rS}(T_{(n)})} =
\prod_{i<j}\frac{2\pi}{\Inner{\alpha_{ij}}{\varpi}}=\prod_{i<j}\frac{2\pi}{j-i}
= \frac{(2\pi)^{\binom{n}{2}}}{\prod^n_{j=1}\Gamma(j)}.
\end{eqnarray}
A direct computation yields (see, for instance, \cite{Zhang2015})
\begin{eqnarray}
\vol_{\rH\rS}(T_{(n)}) = (2\pi)^n.
\end{eqnarray}
Then
\begin{eqnarray}
\vol_{\rH\rS}(\rU(n)) =
\frac{(2\pi)^{\frac{n(n+1)}{2}}}{\prod^n_{j=1}\Gamma(j)}.
\end{eqnarray}
\end{exam}

We now move on to describing the roots of $\k$.
Let $\Phi_{\u(m)}$ denote the set of roots of a Lie algebra $\u(m)$.
For each $\alpha^{(m)}\in \Phi_{\u(m)}$, denote its associated root
space by $\u(m)_{\alpha^{(m)}}$. Let
$\tau_m(\cdot):=\frac1m\Tr{\cdot}$ be the normalized trace form on
the $m\times m$ matrices so that $\tau_m(\I_m)=1$. For any
$\alpha^{(m)}\in\Phi_{\u(m)}$, the tensor product $\alpha^{(m)}\ot
\tau_n$ is a purely imaginary-valued $\real$-linear map
$\liet_{(m)}\ot \I_n\to\complex$. Extend the domain of
$\alpha^{(m)}\ot\tau_n$ by zero  to $\liet=\liet_{(m)}\ot \I_n + \I_m\ot
\liet_{(n)}$ (recall that the intersection $(\liet_{(m)}\ot
\I_n)\cap(\I_m\ot \liet_{(n)})=\sqrt{-1}\real\cdot \I_m\ot\I_n$, and
on such intersection, $\alpha^{(m)}\ot\tau_n$ yields the value $0$
because $\alpha^{(m)}$ depends only on the differences in the
diagonal entries). We denote this extension by $\tilde \alpha^{(m)}=
\alpha^{(m)}\ot\tau_n$. Symmetrically, for each root $\alpha^{(n)}$
of $\u(n)$, we denote by $\tilde\alpha^{(n)}$ the purely
imaginary-valued $\real$-linear map $\liet\to\complex$ obtained by
extending $\tau_m\ot \alpha^{(n)}$ by zero.

\begin{prop}\label{prop:tsr-rtsp}
Let $\widetilde \Phi_{\u(m)}=\Set{\tilde
\alpha^{(m)}\mid\alpha^{(m)}\in\Phi_{\u(m)}}$ and $\widetilde
\Phi_{\u(n)}=\Set{\tilde
\alpha^{(n)}\mid\alpha^{(n)}\in\Phi_{\u(n)}}$. Their disjoint union
yields the set of roots for $\k=\u(m)\ot\I_n+\I_m\ot\u(n)$:
\begin{eqnarray}
\Phi_\k =\widetilde \Phi_{\u(m)}\bigsqcup\widetilde \Phi_{\u(n)}.
\end{eqnarray}
The root space associated with
$\tilde \alpha^{(m)}\in \widetilde \Phi_{\u(m)}$ is $\tilde
\u(m):=\u(m)\ot\I_n$. The root space associated with $\tilde
\alpha^{(n)}\in \widetilde \Phi_{\u(n)}$ is $\tilde
\u(n):=\I_m\ot\u(n)$.
\end{prop}

\begin{proof}
Take the root space decompositions
(the subscript $\complex$ denotes  complexification)
\begin{eqnarray}
\u(m)_\complex &=& \liet_{(m),\complex}\oplus
\biggl(\bigoplus^m_{i\neq j}\u(m)_{\alpha^{(m)}_{ij}}\biggr),\\
\u(n)_\complex &=& \liet_{(n),\complex}\oplus
\biggl(\bigoplus^n_{k\neq l}\u(n)_{\alpha^{(n)}_{kl}}\biggr).
\end{eqnarray}
Then,
\begin{eqnarray}
&&\k_\complex = \u(m)_\complex\ot\I_n+\I_m\ot\u(n)_\complex\\
&&=(\liet_{(m),\complex}\ot\I_n+\I_m\ot\liet_{(n),\complex}) +
\biggl(\bigoplus^m_{i\neq j}\u(m)_{\alpha^{(m)}_{ij}}\ot\I_n\biggr)+
\biggl(\bigoplus^n_{k\neq l}\I_m\ot\u(n)_{\alpha^{(n)}_{kl}}\biggr).
\end{eqnarray}
We need to show that $\u(m)_{\alpha^{(m)}_{ij}}\ot\I_n$ and $\I_m\ot\u(n)_{\alpha^{(n)}_{kl}}$ are root spaces for $\k$. Take
an arbitrary vector $Z\in
\liet_{(m),\complex}\ot\I_n+\I_m\ot\liet_{(n),\complex}$; then
$$
Z=X\ot\I_n+\I_m\ot Y
$$
for some $X\in \u(m)_\complex$ and $Y\in\u(n)_\complex$. Observe
that
\begin{eqnarray}
[Z,\u(m)_{\alpha^{(m)}_{ij}}\ot\I_n] =
[X,\u(m)_{\alpha^{(m)}_{ij}}]\ot\I_n
=\alpha^{(m)}_{ij}(X)(\u(m)_{\alpha^{(m)}_{ij}}\ot\I_n).
\end{eqnarray}
So $\u(m)_{\alpha^{(m)}_{ij}}\ot\I_n$ is indeed a root space for
$\k$. Similar argument proves that
$\I_m\ot\u(n)_{\alpha^{(n)}_{kl}}$ is a root space for $\k$.

We have demonstrated so far that $\Phi_\k = \widetilde
\Phi_{\u(m)}\cup\widetilde\Phi_{\u(n)}$. To show that this is a
disjoint union, we claim that, for any pair
$(\alpha^{(m)},\alpha^{(n)})\in \Phi_{\u(m)}\times\Phi_{\u(n)}$, the
two root spaces $\u(m)_{\alpha^{(m)}}\ot\I_n$ and
$\I_m\ot\u(n)_{\alpha^{(n)}}$ are distinct. Since root spaces are
$1$-dimensional, our claim is equivalent to saying that
\begin{eqnarray}
\Pa{\u(m)_{\alpha^{(m)}}\ot\I_n} \cap
\Pa{\I_m\ot\u(n)_{\alpha^{(n)}}} = \set{0}.
\end{eqnarray}
To prove the above equation, say $X\in \u(m)_{\alpha^{(m)}}$ and
$Y\in \u(n)_{\alpha^{(n)}}$. As root vectors, the nonzero entries of
$X$ and $Y$ are all off-diagonal. Now suppose $X\ot\I_n=\I_m\ot Y$.
Because all diagonal entries of both $X$ and $Y$ are zero, then it
is easy to see that both $X$ and $Y$ are zero matrices. This
completes the proof.
\end{proof}
With the above Proposition~\ref{prop:tsr-rtsp} at hand,
we can calculate
 $\vol_{\rH\rS}(K)$ for $K=\rU(m)\ot\rU(n)$
using Harish-Chandra's formula:
\begin{eqnarray}
&&\vol_{\rH\rS}(\rU(m)\ot\rU(n)) = \vol_{\rH\rS}(T_{(m)}\ot
T_{(n)})\prod_{\alpha\in\Phi^+_\k}\frac{2\pi}{\Inner{\alpha}{\varpi}}\\
&&=
(2\pi)^{m+n-1}\prod_{\alpha^{(m)}\in\Phi^+_{\u(m)}}\frac{2\pi}{\Inner{\tilde
\alpha^{(m)}}{\widetilde \varpi_{\u(m)}+\widetilde
\varpi_{\u(n)}}}\prod_{\alpha^{(n)}\in\Phi^+_{\u(n)}}\frac{2\pi}{\Inner{\tilde
\alpha^{(n)}}{\widetilde \varpi_{\u(m)}+\widetilde \varpi_{\u(n)}}}
\end{eqnarray}
where $\Phi^+_\k$ is the set of all positive roots for $\k$.
Note that
\begin{eqnarray}
\vol_{\rH\rS}(T_{(m)}\ot T_{(n)}) = (2\pi)^{m+n-1}
\end{eqnarray}
since $\dim(\liet_{(m)}\ot\I_n+\I_m\ot\liet_{(n)})=m+n-1$.
Following the conventions in Example \ref{ex:unvol}, we have
\begin{eqnarray*}
&&\Inner{\tilde \alpha^{(m)}_{ij}}{\widetilde
\varpi_{\u(m)}+\widetilde \varpi_{\u(n)}} =
\Inner{\alpha^{(m)}_{ij}\ot\tau_n}{\varpi_{\u(m)}\ot\tau_n}+\Inner{\alpha^{(m)}_{ij}\ot\tau_n}{\tau_m\ot\varpi_{\u(n)}}\\
&&=\Inner{\alpha^{(m)}_{ij}}{\varpi_{\u(m)}}\Inner{\tau_n}{\tau_n}+\Inner{\alpha^{(m)}_{ij}}{\tau_m}\Inner{\tau_n}{\varpi_{\u(n)}},
\end{eqnarray*}
that is,
\begin{eqnarray}
\Inner{\tilde \alpha^{(m)}_{ij}}{\widetilde
\varpi_{\u(m)}+\widetilde \varpi_{\u(n)}}=\frac{j-i}n.
\end{eqnarray}
Similarly, we have
\begin{eqnarray}
\Inner{\tilde \alpha^{(n)}_{kl}}{\widetilde
\varpi_{\u(m)}+\widetilde \varpi_{\u(n)}}=\frac{l-k}m.
\end{eqnarray}
Furthermore,
\begin{eqnarray}
\prod_{\alpha^{(m)}\in\Phi^+_{\u(m)}}\frac{2\pi}{\Inner{\tilde
\alpha^{(m)}}{\widetilde \varpi_{\u(m)}+\widetilde \varpi_{\u(n)}}}
&=& \prod_{1\leqslant i<j\leqslant m}\frac{2\pi}{\Inner{\tilde
\alpha^{(m)}_{ij}}{\widetilde \varpi_{\u(m)}+\widetilde
\varpi_{\u(n)}}}\\
&=& \prod_{1\leqslant i<j\leqslant m}\frac{2n\pi}{j-i} =
n^{\binom{m}{2}}\vol_{\rH\rS}(\rU(m)/T_{(m)})
\end{eqnarray}
and
\begin{eqnarray}
\prod_{\alpha^{(n)}\in\Phi^+_{\u(n)}}\frac{2\pi}{\Inner{\tilde
\alpha^{(n)}}{\widetilde \varpi_{\u(m)}+\widetilde \varpi_{\u(n)}}}
&=& \prod_{1\leqslant k<l\leqslant n}\frac{2\pi}{\Inner{\tilde
\alpha^{(n)}_{kl}}{\widetilde \varpi_{\u(m)}+\widetilde
\varpi_{\u(n)}}}\\
&=& \prod_{1\leqslant k<l\leqslant n}\frac{2m\pi}{l-k} =
m^{\binom{n}{2}}\vol_{\rH\rS}(\rU(n)/T_{(n)}).
\end{eqnarray}
Therefore we can draw the following conclusion:
\begin{thrm}\label{th:vol-tensor-group}
For positive integers $m$ and $n$ greater than $1$, we have the following
volume formulas:
\begin{eqnarray}
\vol_{\rH\rS}(\rU(m)\ot\rU(n)/T_{(m)}\ot T_{(n)}) &=&
m^{\binom{n}{2}}n^{\binom{m}{2}}\vol_{\rH\rS}(\rU(m)/T_{(m)})\vol_{\rH\rS}(\rU(n)/T_{(n)})\\
&=&m^{\binom{n}{2}}n^{\binom{m}{2}}\frac{(2\pi)^{\binom{m}{2}+\binom{n}{2}}}{\prod^m_{i=1}\Gamma(i)\prod^n_{j=1}\Gamma(j)}, \\
\vol_{\rH\rS}(T_{(m)}\ot T_{(n)}) &=& (2\pi)^{m+n-1},\\
\vol_{\rH\rS}\Pa{\rU(m)\ot\rU(n)} &=&
m^{\binom{n}{2}}n^{\binom{m}{2}}\frac{(2\pi)^{\binom{m+1}{2}+\binom{n+1}{2}-1}}{\prod^m_{i=1}\Gamma(i)\prod^n_{j=1}\Gamma(j)}.
\end{eqnarray}
In particular, for $(m,n)=(2,2)$, we have
\begin{eqnarray}
\vol_{\rH\rS}\Pa{\rU(2)\ot\rU(2)} = 128\pi^5.
\end{eqnarray}
\end{thrm}

\subsection{Symplectic volume}

The relation between flag manifolds and coadjoint orbits is
well-known. To wit, let $\lambda\in \sqrt{-1}\liet^*\simeq \real^n$
and assume that $\lambda$ is dominant and regular under the
coadjoint action. Let $\cO_{K,\lambda}$ denote the orbit $\lambda$.
Then there is a $K$-equivariantly diffeomorphism
\begin{eqnarray*}
\cO_{K,\lambda} \simeq K/T.
\end{eqnarray*}
Moreover, $\cO_{K,\lambda}$ has a standard symplectic form known as
the Kirillov-Kostant-Souriau form. The following proposition gives a
formula for the symplectic volume of $\cO_{K,\lambda}$ (for a proof,
see \cite[Prop.~7.26]{Berline1992}):

\begin{prop}\label{prop:vol-of-coadjoint-orbit}
Let $K$ be a compact connected Lie group of which $T$ is a maximal
torus. Let $\lambda$ be a dominant vector in $\sqrt{-1}\liet^*$ that
is a regular point under the coadjoint action. Let $\cO_{K,\lambda}$
be the orbit through $\lambda$. The symplectic volume of
$\cO_{K,\lambda}$ relative to the standard symplectic form is
\begin{eqnarray}
\vol_{\mathrm{symp}} (\cO_{K,\lambda}) =
\prod_{\alpha\in\Phi^+_\k}\frac{\Inner{\lambda}{\alpha}}{\Inner{\varpi}{\alpha}}.
\label{eq:coadvol}
\end{eqnarray}
Here $\varpi=\frac12\sum_{\alpha\in\Phi^+_\k}\alpha$.
\end{prop}

\begin{exam}\label{ex:un-symp-vol}
Consider the unitary group $\rU(n)$. Using the Hilbert-Schmidt inner product,
we can identify $\u(n)^*$ and
$\u(n)$; since the Hilbert-Schmidt inner product is
invariant under the adjoint action, we can also identify
adjoint and coadjoint orbits.
So we may speak of the symplectic volume of
the adjoint orbit $\cO_{\Lambda}$ of
$\Lambda\in\sqrt{-1}\liet$ where $\liet$ is the standard maximal toral subalgebra
of $\u(n)$, that is, the set of diagonal matrices in $\u(n)$.
Note that
$$
\cO_\Lambda=\Ad(\rU(n))\Lambda = \Set{U\Lambda U^*:
\Lambda=\diag(\lambda_1,\ldots,\lambda_n)\text{ with }
\lambda_i\in\real\text{ for all $i=1,\dotsc,n$ and }\lambda_1>\cdots>\lambda_n}.
$$
Continuing with the conventions set up in Example~\ref{ex:unvol}, we have, by Equation \eqref{eq:coadvol},
\begin{eqnarray}\label{eq:adjoint-orbit}
\vol_{\mathrm{symp}}(\cO_\Lambda) =
\prod_{i<j}\frac{\Inner{\alpha_{ij}}{\lambda}}{\Inner{\alpha_{ij}}{\varpi}}
= \prod_{i<j}\frac{\lambda_i-\lambda_j}{j-i} =
\frac{\prod_{i<j}(\lambda_i-\lambda_j)}{1!\cdots(n-1)!}
=\frac{\prod_{i<j}(\lambda_i-\lambda_j)} {\prod^n_{j=1}\Gamma(j)}.
\end{eqnarray}
This result appears in \cite[Eq.~(2.3)]{Christandl2014}.
\end{exam}

\begin{thrm}\label{th:Vol-LUO}
Let $K=\rU(m)\ot \rU(n)$ and $T=T_{(m)}\ot T_{(n)}$
as stated at the beginning of this section.
Let $\Lambda = \diag(\lambda_{1,1},\ldots,\lambda_{m,1},
\ldots,\lambda_{1,n},\ldots,\lambda_{m,n})$ where $\lambda_{i,j}$'s are
real numbers satisfying $\lambda_{1,1}>\ldots>\lambda_{m,n}$.
Then $\Lambda$ is a regular point in the
maximal toral subalgebra $\sqrt{-1}\liet_{(mn)}$ of $\sqrt{-1}\u(mn)$.
Let $\lambda$ be a regular point in $\sqrt{-1}\liet$ derived
from $\Lambda$, i.e.,
$\lambda=\ptr{n}{\Lambda}\ot\I_n+\I_m\ot\ptr{m}{\Lambda}$. Then
the symplectic volume of the adjoint orbit $\cO_{\lambda}$ is given by
\begin{eqnarray}
\vol_{\mathrm{symp}} (\cO_{\lambda}) &=&
\vol_{\mathrm{symp}}(\cO_{\ptr{m}{\Lambda}})
\vol_{\mathrm{symp}}(\cO_{\ptr{n}{\Lambda}}) \label{eq:symp-prd}
 \\
 &=& \frac{\bigl[\prod_{1\leqslant i<j\leqslant m}
 \sum^n_{k=1}(\lambda_{ik} -
\lambda_{jk})\bigr] \bigl[\prod_{1\leqslant i<j\leqslant n}
\sum^m_{i=1} (\lambda_{ik} -
\lambda_{il}) \bigr]}{\prod^m_{i=1}\Gamma(i)\prod^n_{k=1}\Gamma(k)}.
\label{eq:tsr-sym-vol}
\end{eqnarray}
\end{thrm}

\begin{proof}
Thanks to Propositions \ref{prop:tsr-rtsp} and \ref{prop:vol-of-coadjoint-orbit},
we have
\begin{eqnarray}
\vol_{\mathrm{symp}} \Pa{\cO_{\lambda}}
=\prod_{\alpha^{(m)}\in\Phi^+_{\u(m)}}\frac{\Inner{\tilde\alpha^{(m)}}{\lambda}}{\Inner{\tilde\alpha^{(m)}}{\widetilde\varpi_{\u(m)}+\widetilde
\varpi_{\u(n)}}}\prod_{\alpha^{(n)}\in\Phi^+_{\u(n)}}\frac{\Inner{\tilde\alpha^{(n)}}{\lambda}}{\Inner{\tilde\alpha^{(n)}}{\widetilde\varpi_{\u(m)}+\widetilde
\varpi_{\u(n)}}}.
\end{eqnarray}
Continuing with the conventions used in Example \ref{ex:un-symp-vol}, wea hve
\begin{eqnarray}
&&\vol_{\mathrm{symp}} \Pa{\cO_{\lambda}} =\prod_{1\leqslant
i<j\leqslant
m}\frac{\Inner{\tilde\alpha^{(m)}_{ij}}{\lambda}}{\Inner{\tilde\alpha^{(m)}_{ij}}{\widetilde\varpi_{\u(m)}+\widetilde
\varpi_{\u(n)}}}\prod_{1\leqslant i<j\leqslant
n}\frac{\Inner{\tilde\alpha^{(n)}_{kl}}{\lambda}}{\Inner{\tilde\alpha^{(n)}_{kl}}{\widetilde\varpi_{\u(m)}+\widetilde
\varpi_{\u(n)}}}.
\end{eqnarray}
Since
\begin{eqnarray}
\Inner{\tilde\alpha^{(m)}_{ij}}{\lambda} =
\Inner{\alpha^{(m)}_{ij}\ot\tau_n}{\lambda} =
\frac1n\Tr{((E_{ii}-E_{jj})\ot\I_n)\Lambda} =
\frac1n\Tr{(E_{ii}-E_{jj})\Ptr{n}{\Lambda}}
\end{eqnarray}
and
\begin{eqnarray}
\Inner{\tilde\alpha^{(n)}_{kl}}{\lambda} =
\Inner{\tau_m\ot\alpha^{(n)}_{kl}}{\lambda} =
\frac1m\Tr{(\I_m\ot(E_{kk}-E_{ll}))\Lambda}=
\frac1m\Tr{(E_{kk}-E_{ll})\Ptr{m}{\Lambda}}
\end{eqnarray}
it follows that
\begin{eqnarray}
\Inner{\tilde\alpha^{(m)}_{ij}}{\lambda} =\frac1n
\sum^n_{k=1}\Pa{\lambda_{ik} - \lambda_{jk}}
\end{eqnarray}
and
\begin{eqnarray}
\Inner{\tilde\alpha^{(n)}_{kl}}{\lambda}
=\frac1m\sum^m_{i=1}\Pa{\lambda_{ik} - \lambda_{il}}.
\end{eqnarray}
Therefore,
\begin{eqnarray}
\vol_{\mathrm{symp}} \Pa{\cO_{\lambda}} &=&\prod_{1\leqslant
i<j\leqslant m}\frac{\frac1n \sum^n_{k=1}\Pa{\lambda_{ik} -
\lambda_{jk}}}{\frac1n(j-i)}\prod_{1\leqslant i<j\leqslant
n}\frac{\frac1m\sum^m_{i=1}\Pa{\lambda_{ik} -
\lambda_{il}}}{\frac1m(l-k)}\\
&=& \prod_{1\leqslant i<j\leqslant m}\frac{
\sum^n_{k=1}\Pa{\lambda_{ik} - \lambda_{jk}}}{j-i}\prod_{1\leqslant
i<j\leqslant n}\frac{\sum^m_{i=1}\Pa{\lambda_{ik} -
\lambda_{il}}}{l-k}.
\end{eqnarray}
Together with \eqref{eq:adjoint-orbit}, we have the desired
equalities \eqref{eq:tsr-sym-vol} and
 \eqref{eq:symp-prd}.
\end{proof}

\begin{cor}\label{cor:coadjoint-LU(4)}
Let $K=\rU(2)\ot\rU(2)$ and $T=T_{(2)}\ot T_{(2)}$.
Let
$\Lambda=\diag(\lambda_1,\lambda_2,\lambda_3,\lambda_4)\in \sqrt{-1}\mathfrak{t}$
with $\lambda_1>\lambda_2>\lambda_3>\lambda_4\geqslant0$ and
$\sum^4_{j=1}\lambda_j=1$.
Let $\cO_{\lambda}^{\mathrm{LU}}$ denote the adjoint orbit
of $\Lambda$, that is,
$$
\cO^{\mathrm{LU}}_\Lambda:=\Set{(U\ot V)\Lambda (U^{-1}\ot V^{-1}):
U,V\in\rU(2)}.
$$
Then the symplectic volume of such local unitary orbit is given by
\begin{eqnarray}
\vol_{\mathrm{symp}} \Pa{\cO^{\mathrm{LU}}_\lambda} =
(\lambda_1+\lambda_2-\lambda_3-\lambda_4)(\lambda_1+\lambda_3-\lambda_2-\lambda_4).
\end{eqnarray}
\end{cor}

\begin{proof}
This is a direct consequence of Theorem~\ref{th:Vol-LUO}.
\end{proof}

\begin{remark}
Let $\Lambda$ be as in Corollary~\ref{cor:coadjoint-LU(4)},
and let
\begin{eqnarray}
\cO^{\rG\rU}_\Lambda:=\Set{W\Lambda W^{-1}: W\in\rU(4)}.
\end{eqnarray}
By \eqref{eq:adjoint-orbit}, we have
\begin{eqnarray}
\vol_{\mathrm{symp}} (\cO^{\rG\rU}_\Lambda)=
\frac1{12}\prod_{1\leqslant i<j\leqslant 4}(\lambda_i-\lambda_j).
\end{eqnarray}
Meanwhile, under the constraint
$\lambda_1>\lambda_2>\lambda_3>\lambda_4\geqslant0$ and
$\sum^4_{j=1}\lambda_j=1$, we have
\begin{eqnarray*}
&&\vol_{\mathrm{symp}} (\cO^{\mathrm{LU}}_\lambda) =
[(\lambda_1-\lambda_3)+(\lambda_2-\lambda_4)][(\lambda_1-\lambda_2)+(\lambda_3-\lambda_4)]\\
&&>
4\sqrt{(\lambda_1-\lambda_3)(\lambda_2-\lambda_4)(\lambda_1-\lambda_2)(\lambda_3-\lambda_4)}\\
&&>
4(\lambda_1-\lambda_3)(\lambda_2-\lambda_4)(\lambda_1-\lambda_2)(\lambda_3-\lambda_4)\\
&&> 4\prod_{1\leqslant i<j\leqslant 4} (\lambda_i-\lambda_j)>48
\vol_{\mathrm{symp}} (\cO^{\rG\rU}_\Lambda)
> \vol_{\mathrm{symp}}  (\cO^{\rG\rU}_\Lambda).
\end{eqnarray*}
So we see that $\vol_{\mathrm{symp}} (\cO^{\rG\rU}_\Lambda) <
\vol_{\mathrm{symp}}(\cO^{\mathrm{LU}}_\lambda)$. At the same time,
$\cO^{\rL\rU}_\Lambda$ is a submanifold of $\cO^{\rG\rU}_\Lambda$
since $\rU(2)\ot\rU(2)$ is a Lie subgroup of $\rU(4)$; but this is
not a contradiction because the measures for
$\cO^{\mathrm{LU}}_\Lambda$ and $\cO^{\rG\rU}_\Lambda$ have no
\emph{a priori} relation, so one cannot directly compare the two. In
fact, $\cO^{\rL\rU}_\Lambda$ is a set of zero-measure in
$\cO^{\rG\rU}_\Lambda$ because
$\dim(\cO^{\rL\rU}_\Lambda)<\dim(\cO^{\rG\rU}_\Lambda)$.
\end{remark}

\section{Discussion and concluding remarks}\label{sect:discuz-&-conclusion}

There may exist many $\Lambda$'s corresponding to
a single $\lambda$ in Theorem~\ref{th:Vol-LUO}.
A relevant well-known problem is the so-called quantum marginal problem.
For the two-qubit system, there is a nice solution for it
\cite{Bravyi2004}. Specifically, mixed two-qubit state $\rho_{AB}$
with spectrum $\lambda_1\geqslant \lambda_2\geqslant
\lambda_3\geqslant \lambda_4\geqslant0$ and margins $\rho_A$ and $\rho_B$
exists if and only if minimal eigenvalues $\lambda_A$ and $\lambda_B$ of
the margins satisfy the following inequalities:
\begin{eqnarray}
\begin{cases}
\min(\lambda_A,\lambda_B)\geqslant \lambda_3+\lambda_4,\\
\lambda_A+\lambda_B\geqslant \lambda_2+2\lambda_3+\lambda_4,\\
\abs{\lambda_A-\lambda_B}\leqslant
\min(\lambda_1-\lambda_3,\lambda_2-\lambda_4).
\end{cases}
\end{eqnarray}
Here we examine a specific example showing this property. Let
$\rho_{AB}$ be any two-qudit in $\density{\complex^d\ot\complex^d}$.
Then there exists a global unitary $V\in\rU(d^2)$ such that
$\rho'_{AB} = V\rho_{AB} V^\dagger$ with two marginal states as
$\rho'_A = \rho'_B = \I_d/d$. Indeed, by the Spectral Decomposition
Theorem, we have the following decomposition: Writing
$[k]:=\set{1,\ldots,k}$ for any positive integer $k$, we have
\begin{eqnarray}
\rho_{AB} = \sum^{d^2}_{j=1}\lambda_j \out{\Psi_j}{\Psi_j},\quad
\ket{\Psi_j}\in \complex^d\ot\complex^d,
\end{eqnarray}
where $\lambda_j\geqslant0$ for each $j\in[d^2]$, and
$\Set{\ket{\Psi_j}:j\in[d^2]}$ are the eigenvectors corresponding to
eigenvalues $\lambda_j$. There exists a collection of unitary
matrices, called \emph{discrete Weyl unitary matrices},
$W_j\in\rU(d)$, $j\in[d^2]$, that form a unitary matrix basis for
$M_d(\complex)$, the set of all $d\times d$ complex matrices.
If we denote by $\vec(M)$ is the
vectorization of a complex rectangular matrix $M$, that is,
 $\vec(M) := \sum_{i,j}M_{ij}\ket{ij}$ where
$M=\sum_{i,j}M_{ij}\out{i}{j}$, then
$\Set{\vec(W_j):j\in[d^2]}$ forms a maximally entangled basis for
$\complex^d\ot\complex^d$ (see also for its generalization in
\cite{Guo2015}).
So there is a global unitary matrix
$V\in \rU(d^2)$ such that
\begin{eqnarray}
V\ket{\Psi_j}=\frac1{\sqrt{d}}\vec(W_j), j \in[d^2],
\end{eqnarray}
since $\Set{\ket{\Psi_j}:j\in[d^2]}$ and
$\Set{\frac1{\sqrt{d}}\vec(W_j):j\in[d^2]}$ are two orthonormal
bases for the same space $\complex^d\ot\complex^d$. This implies that
\begin{eqnarray}
\rho'_{AB}=V\rho_{AB}V^\dagger = \sum^{d^2}_{j=1}\lambda_j
V\out{\Psi_j}{\Psi_j}V^\dagger =
\frac1d\sum^{d^2}_{j=1}\lambda_j\vec(W_j)\vec(W_j)^\dagger.
\end{eqnarray}
Hence,
\begin{eqnarray}
\rho'_A = \frac1d\sum^{d^2}_{j=1}\lambda_jW_jW_j^\dagger =
\frac{\I_d}d\sum^{d^2}_{j=1}\lambda_j =
\frac{\I_d}d\Tr{\rho_{AB}}=\frac{\I_d}d
\end{eqnarray}
and
\begin{eqnarray}
\rho'_B = \frac1d\sum^{d^2}_{j=1}\lambda_j(W^\dagger_jW_j)^\t =
\frac{\I_d}d\sum^{d^2}_{j=1}\lambda_j =
\frac{\I_d}d\Tr{\rho_{AB}}=\frac{\I_d}d.
\end{eqnarray}
Therefore, $\rho'_A = \rho'_B = \I_d/d$. This example also indicates
that the maximum of mutual information along a global unitary orbit
of a given bipartite state with the prescribed spectrum $\Lambda$ is
$2\ln(d)-S(\Lambda)$ \cite{Jevtic2012a,Jevtic2012b}, where
$S(\Lambda)$ is the von~Neumann entropy. Now fix $\rho_A$ and
$\rho_B$, and denote by $\cC(\rho_A,\rho_B)$ the set of all
bipartite states $\rho_{AB}$ with fixed marginal states $\rho_A$ and
$\rho_B$, respectively. It is known that $\cC(\rho_A,\rho_B)$ is a
compact convex set. Moreover, in Parthasarathy \cite{Patha2005}, a
necessary and sufficient condition is presented for an element
$\rho_{AB}$ in $\cC(\rho_A,\rho_B)$ to be an extreme point; in the
two-qubit case, the condition amounts to a two-qubit state
$\rho_{AB}\in \cC(\I/2,\I/2)$ being maximally entangled.

There are several \emph{open} questions which are presented below:
\begin{question}
Suppose that bipartite states $\rho$ and $\tilde\rho$ are not LU
equivalent, while $\rho$ is locally diagonalizable but $\tilde\rho$
is not. As we have seen previously, the HS volume of
$\cO^{\rL\rU}_\rho$ can be calculated. The question is how to
calculate the HS volume of $\cO^{\rL\rU}_{\tilde\rho}$?
\end{question}

\begin{question}
It is easily seen that $\sD_{\rL\rU}(\complex^m\ot\complex^n)$ can
be partitioned into local unitary orbits, that is,
\begin{eqnarray}
\sD_{\rL\rU}(\complex^m\ot\complex^n) =
\bigsqcup_{\Lambda\in\Delta_{mn-1}}\cO^{\rL\rU}_\Lambda.
\end{eqnarray}
How do we get the HS volume
$\vol_{\rH\rS}\Pa{\sD_{\rL\rU}(\complex^m\ot\complex^n)}$ from the
HS volumes $\vol_{\rH\rS}(\cO^{\rL\rU}_\Lambda)$ of local unitary
orbits $\cO^{\rL\rU}_\Lambda$?
\end{question}

\begin{question}
Can we establish some kind of a "canonical form" for bipartite
states like the Singular Value Decomposition (SVD) for
complex matrices and/or the Spectral Decomposition for normal
matrices, especially for any two-qubit states? Also, how to obtain that?
\end{question}

\begin{question}
Because $\cC(\rho_A,\rho_B)$ is a compact convex set, its HS volume
can in principle be calculated. To the best of our knowledge, an
analytical expression for its HS volume has not been founded. We
note that the volume for $\cC(\I/2,\I/2)$ relative to the Lebesgue
measure can be found in \cite{Lovas2016a}: $\vol_g(\cC(\I/2,\I/2)) =
{2\pi^4}/{315}$. Note also that we can identify $\cC(\I/2,\I/2)$
with the set of all unital qubit quantum channels via
Choi-Jami{\l}kowski isomorphism. The problem of calculating the HS
volume of $\cC(\I/2,\I/2)$ and some relevant discussions can also be
found in \cite{Zhang2015}.
\end{question}

In summary, we analyzed in this paper the geometry of locally
diagonalizable bipartite states its Hilbert-Schmidt volume and
symplectic volume. We obtained an expression for the HS volume
involving an integral, so in principle, we could work out an
analytical formula; but for now, we evaluated the volume in the
two-qubit case. In addition, we obtained a necessary and sufficient
condition for a two-qubit state to be LU equivalent to a diagonal
state. A generalization to higher dimensional cases is still open.
After introducing Harish-Chandra's volume formula for flag
manifolds, we turned to the geometry of local unitary orbits. We
found that Harish-Chandra's volume formula can be applied to
calculate the volume of local unitary orbits. We also obtained a
volume formula for the tensor product $\rU(m)\ot\rU(n)$. Although
this is a direct consequence of Harish-Chandra's volume formula,
there is, to our knowledge, no record of it in the literature. We
believe these results and the questions raised can shed new lights
and spur relevant research in quantum information theory.

\subsubsection*{Acknowledgements}
This research was supported by Zhejiang Provincial Natural Science
Foundation of China under Grant No. LY17A010027 and NSFC
(Nos.11301124, 61771174), and also supported by the
crossdisciplinary innovation team building project of Hangzhou
Dianzi University. SH was partially supported under Northwestern
Scholarship Grant.


\end{document}